\documentclass{article}
\usepackage{geometry}
\usepackage[utf8]{inputenc}
\usepackage{url}
\usepackage{amsmath}
\usepackage{graphicx}
\usepackage{color}
\usepackage{amssymb}
\usepackage{amsthm}
\usepackage{paralist}
\usepackage{microtype}
\usepackage{authblk}
\usepackage[round]{natbib}

\newcommand{\E}{\mathrm{E}}
\newcommand{\R}{\mathbb{R}}
\newcommand{\M}{\mathcal{M}}

\newcommand{\Xpri}{X_{\text{priv}}}
\newcommand{\Ypri}{Y_{\text{priv}}}
\newcommand{\Xpub}{X_{\text{pub}}}
\newcommand{\Ypub}{Y_{\text{pub}}'}
\newcommand{\Zpri}{Z_{\text{priv}}}

\begin{document}
\title{Representation Transfer for Differentially Private Drug Sensitivity Prediction}
\date{}
\author{Teppo Niinim\"aki\,$^{1}$, Mikko Heikkil\"a\,$^{2}$, Antti Honkela\,$^{2,3,4,}$\footnote{These authors jointly supervised the work.}~, and Samuel Kaski\,$^{1,*}$\\

  $^{1}$Helsinki Institute for Information Technology HIIT, Department of Computer Science, Aalto University, Finland, \\
$^{2}$Helsinki Institute for Information Technology HIIT, Department of Mathematics and Statistics, University of Helsinki, Finland, \\
$^{3}$Department of Public Health, University of Helsinki, Finland and \\
$^{4}$Helsinki Institute for Information Technology HIIT, Department of Computer Science, University of Helsinki, Finland.}

\maketitle

\abstract{
  \noindent\textbf{Motivation:} Human genomic datasets often contain sensitive
  information that limits use and sharing of the data. In particular,
  simple anonymisation strategies fail to provide sufficient level of
  protection for genomic data, because the data are inherently
  identifiable. Differentially private machine learning
  can help by guaranteeing that the published results
  do not leak too much information about any individual data point.
  Recent research has reached promising results on differentially
  private drug sensitivity prediction using gene expression data.
  Differentially private learning with genomic data is challenging
  because it is more difficult to guarantee the privacy in high
  dimensions. Dimensionality reduction can help, but if the dimension
  reduction mapping is learned from the data, then it needs to be
  differentially private too, which can carry a significant privacy
  cost. Furthermore, the selection of any hyperparameters (such as the
  target dimensionality) needs to also avoid leaking private
  information.\\
  \textbf{Results:} We study an approach that uses a large public
  dataset of similar type to learn a compact representation for
  differentially private learning. We compare three representation
  learning methods: variational autoencoders, PCA and random
  projection. We solve two machine learning tasks on gene expression
  of cancer cell lines: cancer type classification, and drug sensitivity
  prediction. The experiments demonstrate significant benefit from all
  representation learning methods with variational autoencoders
  providing the most accurate predictions most often. Our results
  significantly improve over previous state-of-the-art in accuracy of
  differentially private drug sensitivity prediction.\\
}

\section{Introduction}

Privacy-preserving machine learning has the potential to enable the
research use of many sensitive datasets, that would otherwise be
out of reach for the community. This is especially the case for
medical data, which almost always contain sensitive information traceable
back to the data subjects. As an example, it has been shown that
individuals can be identified from genomic data \citep{gymrek2013}
including mixtures from several individuals \citep{homer2008}. It is
likely that functional genomics data such as gene expression data are
also identifiable. While different anonymisation strategies
\citep{sweeney2002, machanavajjhala2007, li2007} can protect the
privacy of the data subjects to some degree, they do not have formal
guarantees and can fail to provide sufficient protection in practice
\citep{ganta2008}.

\emph{Differential privacy} (DP) \citep{dwork2006, dwork2014} is a framework that guarantees strict bounds for the amount of leaked private information, even in the presence of arbitrary side information.
The guarantees are obtained by adding specific forms of randomisation
to the computation process. In a machine learning context this usually
means adding noise either directly to the input of the algorithm
(\emph{input perturbation}), to the output (\emph{output
  perturbation}), or modifying the algorithm itself, for instance, by
perturbing the optimisation objective (\emph{objective perturbation}).

The privacy guarantee is controlled by a ``privacy budget'' parameter,
usually denoted by $\epsilon > 0$; smaller $\epsilon$ means stricter
guarantees, and can be achieved by increasing the amount of
noise. Formally, a randomised mechanism $\M$ is said to be
\emph{$\epsilon$-differentially private}, if for all pairs of
\emph{neighbouring} datasets $X, X'$ differing\footnote{There are two
  slightly different definitions of neighbouring datasets. In
  \emph{bounded} case, the value of one sample is allowed to
  change. In \emph{unbounded} case, the addition or removal of one
  sample is allowed. Unbounded $\epsilon$-DP guarantee implies bounded
  $2\epsilon$-DP guarantee. This article uses the bounded case.} on a
single sample and all measurable subsets $S$ of possible outputs,
\[
	\Pr(\M(X) \in S) \le e^\epsilon \Pr(\M(X') \in S).
\]
Intuitively, this means that changing one sample in the dataset can change the output distribution only by a factor $e^\epsilon$.

As an extension, $\M$ is said to be \emph{($\epsilon$, $\delta$)-differentially private}, if
\[
	\Pr(\M(X) \in S) \le e^\epsilon \Pr(\M(X') \in S) + \delta,
\]
for all measurable $S$ and all neighbouring datasets $X, X'$. The
condition with nonzero $\delta > 0$ is often easier to achieve than
pure $\epsilon$-DP.

In this article we are interested in DP learning for drug sensitivity
prediction using gene expression data. First proposed by
\citet{Staunton2001}, the drug sensitivity prediction problem has
attracted significant attention recently, including from a DREAM
challenge in 2012 \citep{Costello2014} that provided standardised
evaluation metrics.  The scale of the cytotoxicity assays needed has
kept the sizes of the available data sets relatively small from a
machine learning perspective. \citet{honkela2018} were the first to
apply DP learning to this problem. They needed to specifically limit
the sensitivity of the learning and the dimensionality of the input
data to make the learning feasible.

In abstract terms, our goal in this problem is DP learning of predictive
models with high-dimensional input data, where both input and output
variables need DP protection.
This is a case where DP methods tend to run into trouble with moderate dataset sizes: 
the amount of noise that needs to be added usually increases quickly with the dimensionality, leading to output that is dominated by the noise. This warrants the use of dimensionality reducing methods with the aim of finding a good low-dimensional representation of the original data. However, unless one uses a random projection or some other ``dummy'' method that does not depend on the data, finding a good representation can also leak private information. For this reason, the dimension reduction method itself would also need to be made differentially private, which can completely invalidate the noise magnitude savings obtained in any downstream task like prediction.

Different solutions have been proposed for various special cases:
\citet{kifer2012} solve sparse linear regression problems by using an
$\epsilon$-DP feature selection algorithm. \citet{honkela2018} utilise
external knowledge to select a relevant subset of
features. \citet{kasiviswanathan2016} show theoretical results on
using random projections to improve DP learning on high-dimensional
problems. Differentially private versions of methods such as PCA
\citep{chaudhuri2012, dwork2014pca} or deep learning \citep{abadi2016,
  acs2018} exist and could be used to learn a representation, but the
noise cost can be impractically large for small but high-dimensional
datasets.

We study a straightforward solution based on feature representation
transfer, similar to self-taught learning of \citet{Raina2007}. By
using an additional non-sensitive dataset to learn the representation,
we can apply more advanced representation learning methods. This
approach has many advantages: we do not need labels for the additional
data set, although in our case we make use of labels for a different
task; and only the main learning algorithm needs to be differentially
private, while the representation can be learned using any non-DP
method.  Additionally, the public data can also be used for optimising
any hyperparameters for the representation learning.  In this article,
we consider PCA and variational autoencoders.

Differentially private transfer learning was recently considered by
\citet{Wang2019} in a hypothesis transfer setting, where models
trained on several related source domains are used to improve learning
in the desired target domain. While related, this approach is not
directly applicable to our problem because the setup is different.

Another related approach was considered by \citet{papernot2017}, who
propose differentially private semi-supervised knowledge transfer that
uses an ensemble of ``teacher'' models trained on private data to
label unlabeled public data, which is then used to train a ``student''
model that will be released. The method is flexible in a sense that it
can use any ``black-box'' model as teachers and student. However, it
is limited to classification tasks. In addition, it seems to require a
rather large private dataset in addition to a small public dataset---a
setting that is somewhat different from what we are interested in.

Yet another strategy is to learn a differentially private unsupervised generative model for the data (including the target variable for the prediction task of interest), use it to generate a synthetic version of the data, and use a non-DP algorithm for the actual learning task of interest. Several methods have been proposed for differentially private data sharing \citep{zhang2017, xie2018, acs2018} that could be used for generative model learning and data generation.
For the problem we are considering, however, this approach is
problematic as it requires solving a more general and difficult
learning task, good solution of which would typically require orders
of magnitude more private data than a direct solution of the original
prediction task.

The rest of this article is organised as follows:
In Section \ref{sec:approach} we formalise the problem setting and give an overview of our proposed approach.
Section \ref{sec:methods} gives more details on the implementation of different parts of the proposed approach.
And finally, in Section \ref{sec:results} we conduct experiments with the approach on two different prediction tasks on genomic data.

\section{Approach}
\label{sec:approach}

We assume a setting where we have a private dataset containing a high-dimensional $n \times d$ feature matrix $\Xpri$ and an $n \times 1$ target vector $\Ypri$, where $n$ is the number of samples and $d$ is the number of features. The goal is to learn a differentially private predictor from $\Xpri$ to $\Ypri$.
As learning to predict $\Ypri$ from high-dimensional $\Xpri$ directly is typically not feasible, with moderate sample size and a reasonable privacy budget, we opt for using public data to learn a low-dimensional representation for $\Xpri$.
Therefore, we also assume a publicly available dataset of an $m \times d$ feature matrix $\Xpub$ and an $m \times 1$ auxiliary target vector $\Ypub$ for a related auxiliary prediction task. While a representation can be learned with $\Xpub$ only, the availability of $\Ypub$ is useful for selecting the size of the representation and any other hyperparameters.

We make the following informal assumptions about the relation of the public and the private data:
(1) $\Xpri$ and $\Xpub$ contain the same set of features and are either draws from the same distribution or otherwise distributed similarly enough that using the same mapping to compute a representation is reasonable.
(2) $\Ypub$ may or may not be of the same type as $\Ypri$, but the prediction tasks should resemble each other enough that the prediction of $\Ypub$ can be used for optimising the hyperparameters for the main task of predicting $\Ypri$.

\begin{figure}[t]
	\centering
	\includegraphics[width=\linewidth]{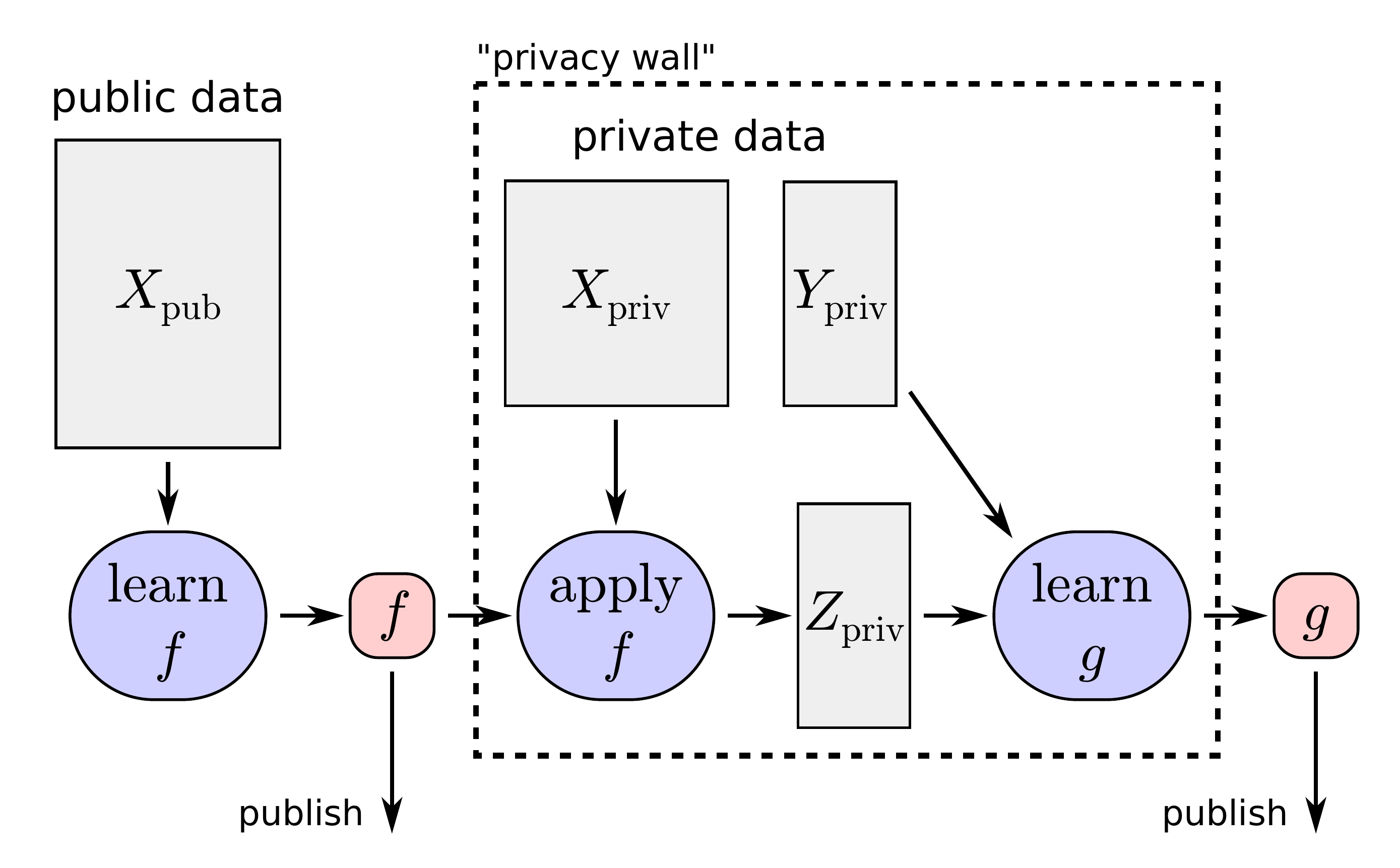}
	\caption{The process of learning $f$ and $g$. Since the learning of $g$ is DP, the leakage of information outside of the ``privacy wall'' is controlled.}
	\label{fig:learning_flow}
\end{figure}

We propose the following procedure:
\begin{enumerate}
	\item Use the public data to learn a dimension-reducing representation mapping $f: \R^d \rightarrow \R^r$, where $r \ll d$, such that $f^{-1}(f(\Xpub)) \approx \Xpub$.
	\item Obtain a low-dimensional representation $\Zpri$ of the private feature data by applying $f$ to $\Xpri$.
	\item Learn a differentially private predictor $g$ such that $g(\Zpri) \approx \Ypri$.
	\item Publish $g \circ f$.
\end{enumerate}
An overview of the learning process is shown in Figure~\ref{fig:learning_flow}.

It is easy to see that the proposed process has the same DP-guarantees as the learning algorithm of step 3:
\newtheorem{theorem}{Theorem}
\begin{theorem}
If step 3 is $(\epsilon, \delta)$-DP w.r.t. $\Zpri$ and $\Ypri$, then the whole process is also $(\epsilon, \delta)$-DP w.r.t $\Xpri$ and $\Ypri$.
\end{theorem}

\begin{proof}
As the learning of $f$ does not use private data, it does not leak any private information. Since each row of $\Zpri$ depends only on the corresponding row of $\Xpri$, $(\epsilon, \delta)$-guarantees w.r.t. $\Zpri$ translate directly to guarantees w.r.t. $\Xpri$.
\end{proof}

In the following section we give some methods that can be used to implement the DP predictor $g$ and the representation mapping $f$. In addition, we describe a procedure for tuning the hyperparameters of $f$.

\section{Methods}
\label{sec:methods}

\subsection{Differentially private prediction}

Later in Section~\ref{sec:results} we will consider prediction tasks that are either real-valued regression or binary classification tasks. Linear regression will be applied to the former and logistic regression to the latter. For now, denote by $X$ the feature matrix and by $y$ the prediction target vector (either real-valued or binary $\{-1, 1\}$)

Logistic regression can be made differentially private with objective
perturbation. The usual non-DP version of the problem can be solved by
minimising the regularised negative log-likelihood
$n^{-1} \sum_{i=1}^n \log(1 + e^{y_i w^T x_i}) + \lambda w^T w$ with
respect to the weight vector $w$, where $x_i$ and $y_i$ denote the
$i$th sample in $X$ and $y$ respectively and $\lambda$ controls the
strength of $L_2$ regularisation. In a method presented by
\citet{chaudhuri2009}, $\epsilon$-DP privacy is obtained by adding a
random bias term $b^T w / n$ (where $b$ is a random vector drawn from
a distribution with density proportional to $e^{-\epsilon ||b|| / 2}$)
to the optimisation objective. The method requires that the samples in
the input feature data are bounded into a 1-sphere.

Like DP logistic regression, also a DP linear regression algorithm can be obtained with an analogous objective perturbation method \citep{kifer2012}. However, since the underlying model belongs to the exponential family, there is also an alternative output-perturbation based $\epsilon$-DP  method that does not require iterative optimisation: Compute the sufficient statistics ($X^T X$, $X^T y$ and $y^T y$) and add noise to them via the Laplace-mechanism \citep{foulds2016}.
We use Bayesian linear regression with sufficient statistic perturbation and data clipping as described by \citet{honkela2018}.

\subsection{Representation learning}

\emph{Random projection} \citep[see e.g.][]{bingham2001} projects the $d$-dimensional data to an $r$-dimensional subspace by multiplying it with a random $d \times r$ projection matrix. This transformation has been shown to preserve approximately the distances between data points \citep{johnson1984}, which is often a desired property for dimensionality reduction methods.

\emph{Principal component analysis} (PCA) finds an orthogonal linear transformation that converts the data to coordinates that are uncorrelated and whose variance decreases from first to last coordinate. When used for dimensionality reduction, only the first $r$ coordinates are kept---these correspond to the $r$ orthogonal directions in which the variance of the original data is the highest.

\emph{Variational autoencoder} (VAE) \citep{kingma2014} learns a generative decoder model $p_\theta(x | z)$, where $z$ is a latent representation of $x$, and an encoder model $q_\xi(z | x)$ that approximates the posterior distribution $p_\theta(z | x)$. Both $p_\theta$ and $q_\xi$ are implemented as neural networks (typically MLPs) and optimised concurrently with variational inference.

We fix $z$ to be low-dimensional, in which case the learned encoder $q_\xi$ can be used for dimensionality reduction by setting $f(x) = \E_{z \sim q_\xi(\cdot | x)}[z]$. (As usual, define $q_\xi(\cdot | x)$ as a multivariate Gaussian distribution parametrised by mean $\mu_\xi(x)$ and diagonal covariance $\Sigma_\xi(x)$, in which case $f(x) = \mu_\xi(x)$.)

\subsection{Optimisation of hyperparameters}

\begin{figure}[t]
	\centerline{\includegraphics[width=\linewidth]{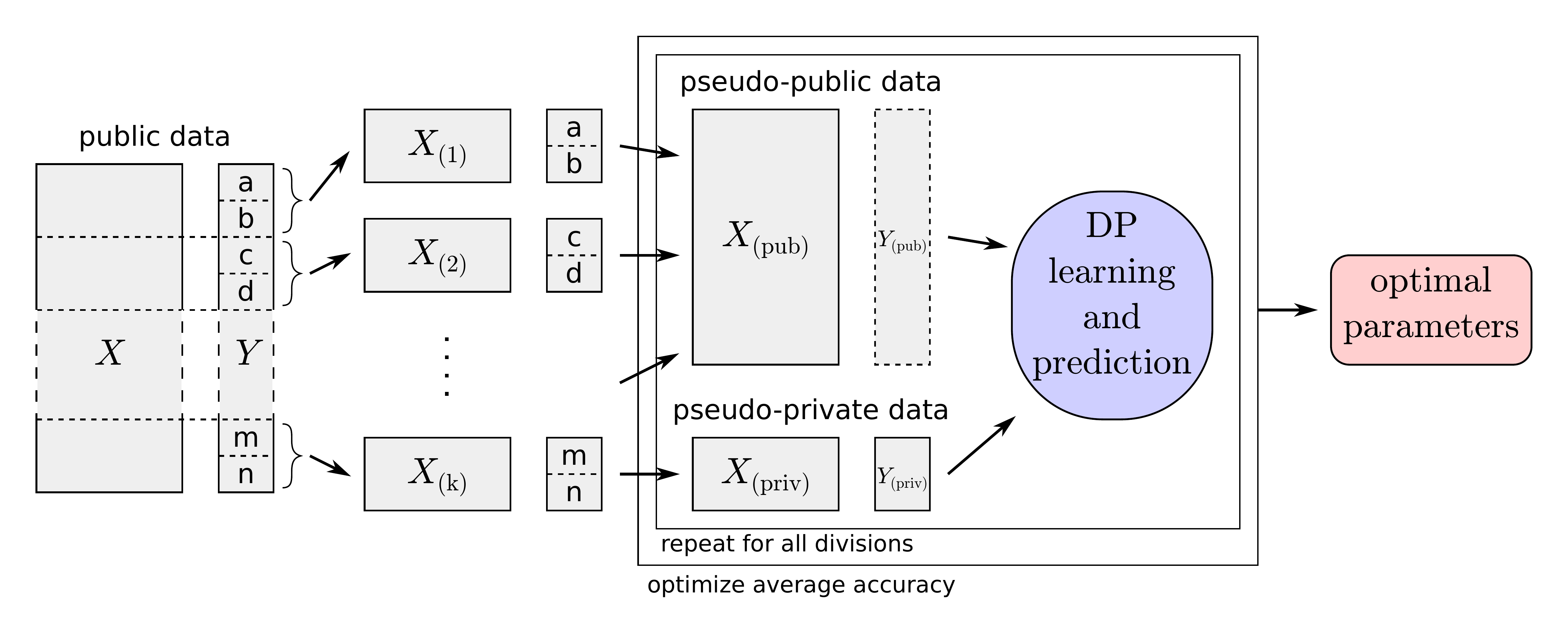}}
	\caption{The process of hyperparameter optimisation. In this example the auxiliary prediction task is assumed to be classification with multiple classes (denoted by a, b, $\ldots$, n), which are partitioned into subsets that consist of two classes each. These are then used in a cross-validation-like hyperparameter optimisation procedure.}
	\label{fig:param_opt}
\end{figure}

For selecting the dimension of the representation and any other hyperparameters of the representation-learning algorithm, we propose a combination of any parameter optimisation approach (such as Bayesian optimisation, random search or grid search) and a cross-validation-like procedure for optimising an auxiliary task of predicting $\Ypub$ from $\Xpub$.
As no private data are used, the parameter optimisation phase does not consume any of the available privacy budget.
In addition, if the auxiliary prediction task uses the same method as the main prediction task, then the hyperparameters could be optimised at the same time.

First the (public) data are divided into $k$ disjoint subsets. Instead
of using one of the subsets as ``validation'' data and the rest as
``training'' data as in cross-validation, we use one of the subsets to
simulate the private data and the rest to simulate the public
data. From now on, these are referred to as \emph{pseudo-private} and
\emph{pseudo-public} sets. The proposed framework (from
Section~\ref{sec:approach}) is then applied to these, that is, a
representation mapping $f$ is learned from the pseudo-public data, $f$
is applied to the features of pseudo-private data, a predictor $g$ is
learned for the pseudo-private target variable and its accuracy is
measured. As in $k$-fold cross-validation, this is repeated for all
$k$ possible selections of the pseudo-private subset. For measuring
the accuracy of $g$, (actual) cross-validation can be used, i.e., the
pseudo-private data can be further divided into different learning and
validation sets.

To mimic the case in which the public and private data do not have
exactly the same distribution, we also want the pseudo-public and
pseudo-private data to be sufficiently different. This guides the
optimiser towards selecting conservative hyperparameters that are more
likely to work well on a wide range of different private datasets. If
the auxiliary prediction task is classification and $\Ypub$ has
multiple classes, the subset division can be based on the classes:
Form each subset by selecting the samples from two (or more) classes.
This strategy is based on the assumption that samples belonging to
different classes have different distributions.  Otherwise, for
instance clustering (based on either $\Xpub$, $\Ypub$ or both) could
be used for finding a good subset division. An overview of the
proposed hyperparameter optimisation method is shown in
Figure~\ref{fig:param_opt}.

\section{Results}
\label{sec:results}

We conducted experiments with two prediction tasks using cancer cell
line gene expression data: cancer type classification and drug
sensitivity prediction.

\subsection{Representation learning for DP cancer type classification}

We first demonstrate the method by classifying TCGA pan-cancer samples
according to the annotated cancer type (e.g.\ lung squamous cell
carcinoma) using RNA-seq gene expression data. In
this task we use the data from The Cancer Genome Atlas (TCGA) project
\citep{tcga} as both the private and public datasets.  We use this
example because it can be performed within the large TCGA
dataset. Because most cancer type pairs are quite easy to identify, we
focus on a number of most difficult pairs.

We used preprocessed TCGA pan-cancer RNA-seq data available at
\url{https://xenabrowser.net/datapages/}.
After further preprocessing (filtering out low-expression genes, applying RLE
normalisation) the dataset contains 10534 samples, 14796 genes and 33
distinct cancer types. We pick two cancer types as private data and
the remaining cancer types form the public dataset. The main and
auxiliary prediction tasks are therefore both cancer type
classification tasks, but for distinct classes. For prediction we use
the differentially private logistic regression algorithm by
\citet{chaudhuri2009}.

\begin{table}
\centering
\caption{The list of cancer type pairs ordered in descending order by the difficulty of classification. The pairs selected to be tested are numbered.}
{\scriptsize
\begin{tabular}{ c l l }
  Case & First cancer type & Second cancer type \\
  \hline 
  1 & lung squamous cell carcinoma & head \& neck squamous cell carcinoma \\
  2 & bladder urothelial carcinoma & cervical \& endocervical cancer \\
  3 & colon adenocarcinoma & rectum adenocarcinoma \\
  4 & stomach adenocarcinoma & esophageal carcinoma \\
  5 & kidney clear cell carcinoma & kidney papillary cell carcinoma \\
  6 & glioblastoma multiforme & sarcoma \\
  - & adrenocortical cancer & uveal melanoma \\
  - & testicular germ cell tumor & uterine carcinosarcoma \\
  - & lung adenocarcinoma & pancreatic adenocarcinoma \\
  7 & ovarian serous cystadenocarcinoma & uterine corpus endometrioid carcinoma \\
  - & brain lower grade glioma & pheochromocytoma \& paraganglioma \\
  - & skin cutaneous melanoma & mesothelioma \\
  - & liver hepatocellular carcinoma & kidney chromophobe \\
  8 & breast invasive carcinoma & prostate adenocarcinoma \\
  - & acute myeloid leukemia & diffuse large B-cell lymphoma \\
  - & thyroid carcinoma & cholangiocarcinoma \\
\end{tabular}
}
\label{tbl:selected_pairs}
\end{table}

\begin{figure}[t]
\centerline{\includegraphics{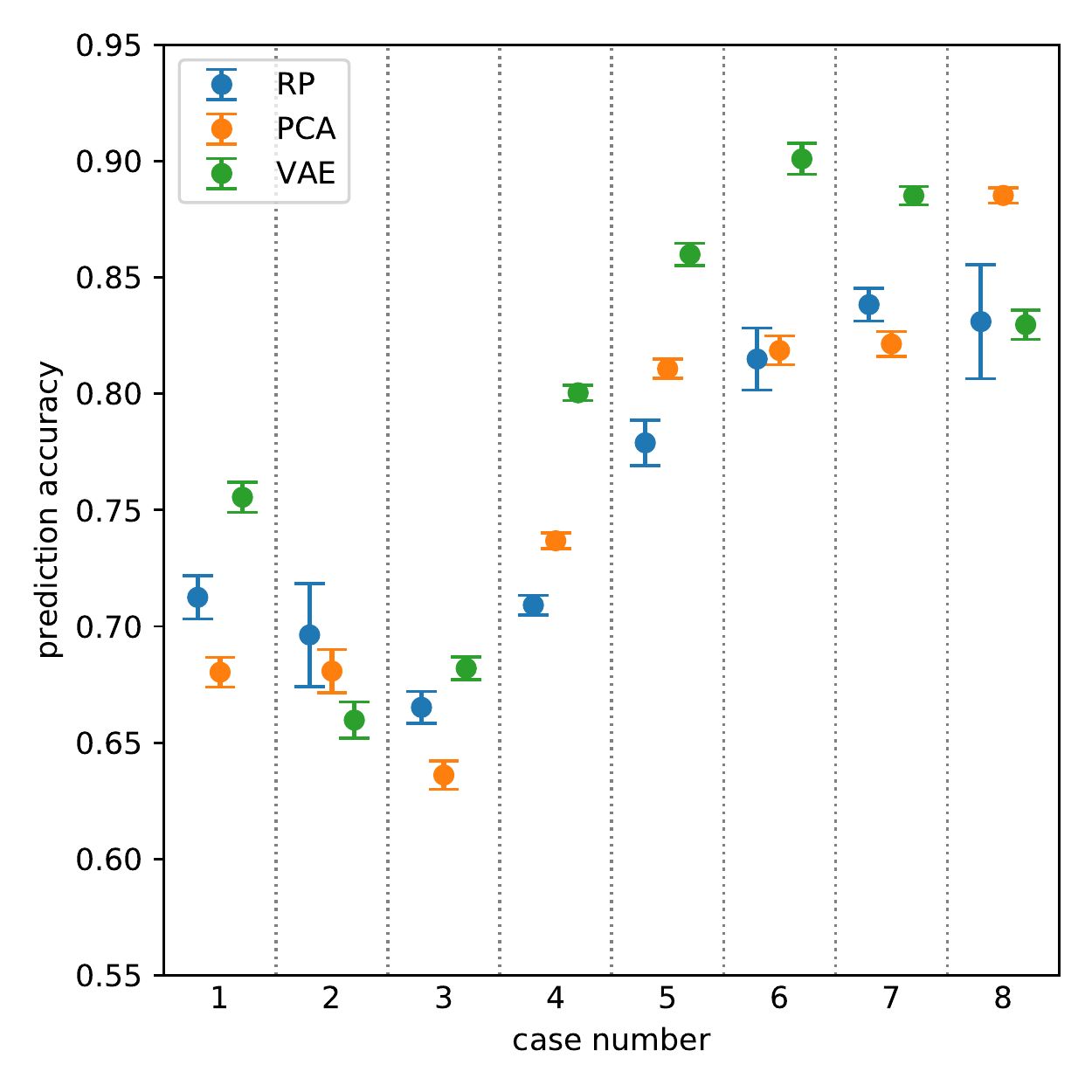}}
\caption{Logistic regression prediction accuracy (the fraction of correctly classified samples) with $\epsilon = 1.0$ in 8 cancer type classification tasks (see Table~\protect\ref{tbl:selected_pairs}). Data: TCGA. Error bars show the standard deviation of the mean accuracy over 9 independent runs of the testing phase.}
\label{fig:cancertype_acc}
\end{figure}

\begin{table}[t]
\centering
\caption{Representation dimensions (repr-dim) and other selected hyperparameters (log learning rate, the number of hidden layers, the size of hidden layers) for different cases on cancer type classification.}
\begin{tabular}{ c | c c c c c c }
  \cline{2-7}
   & \multicolumn{1}{c|}{RP} & \multicolumn{1}{c|}{PCA} & \multicolumn{4}{c|}{VAE} \\
  \cline{2-7}
  case & \multicolumn{3}{c|}{repr-dim} & \multicolumn{1}{c|}{log-lr} & \multicolumn{1}{c|}{layers} & \multicolumn{1}{c|}{layer-dim} \\
  \hline
  1 & 7 & 10 & 5 & -3.5 & 1 & 755 \\
  2 & 7 & 14 & 5 & -4.8 & 1 & 1925 \\
  3 & 8 & 10 & 5 & -5.3 & 2 & 2370 \\
  4 & 8 & 14 & 5 & -5.3 & 2 & 1270 \\
  5 & 7 & 8 & 4 & -4.3 & 1 & 260 \\
  6 & 8 & 10 & 5 & -4.5 & 1 & 330 \\
  7 & 7 & 12 & 5 & -3.7 & 2 & 1510 \\
  8 & 5 & 7 & 4 & -3.8 & 1 & 88 \\
\end{tabular}

\label{tbl:hyperparam}
\end{table}

While the split to private and public data could be done in multiple ways, the prediction task would be quite easy in many of those. Hence, we use the following procedure to produce several of these splits: (1) Consider all $\binom{33}{2}$ possible splits and run a non-DP version of the pipeline (as in Figure~\ref{fig:learning_flow}) with PCA-based reduction to eight-dimensional space. (2) Build a sequence of cancer-type pairs by picking the pair that was the hardest to predict (i.e. has lowest classification accuracy), then from the remaining cancer types again the pair that was hardest, and so on. The result is a sequence of 16 pairs ordered by the prediction difficulty (see Table~\ref{tbl:selected_pairs}). (3) Of these pairs, select the 6 hardest, as well as those 2 of the remaining pairs that had at least 200 samples in both classes.

The full testing pipeline, including the hyperparameter optimisation
phase, was then run separately for each of the 8 selected pairs as a
private dataset. In each case, the remaining 15 pairs form the
$k = 15$ subsets that were used for optimising the hyperparameters.

\paragraph{Methods}
We compare three different representation learning methods: random
projection (RP), principal component analysis (PCA) and variational
autoencoder (VAE) \citep{kingma2014}. VAE was implemented with PyTorch
\citep{paszke2017} and uses 1--3 hidden layers with ReLU activation
functions for both the encoder and the decoder. The learning phase
uses the Adam optimiser \citep{kingma2015} and is given one hour of
GPU time with early stopping. The size of the representation (for RP,
PCA and VAE) and other hyperparameters for VAE (the number of layers,
layer sizes, learning rate) are optimised with GPyOpt
\citep{gpyopt2016}. We also experimented with optimising a much
larger set of hyperparameters, 12 in total, but GPyOpt had
difficulties in obtaining similar levels of performance.

For each of the 8 test cases we ran the hyperparameter optimisation
phase once, giving it 5 days of time. Then with the best found
hyperparameters we ran the final testing 9 times with different RNG
seeds, and report the mean prediction accuracy as well as the standard
deviation of the mean. In measuring the prediction accuracy (both for
hyperparameter optimisation and for final testing) we use 10-fold
cross-validation.

\begin{figure}[t]
\centerline{\includegraphics{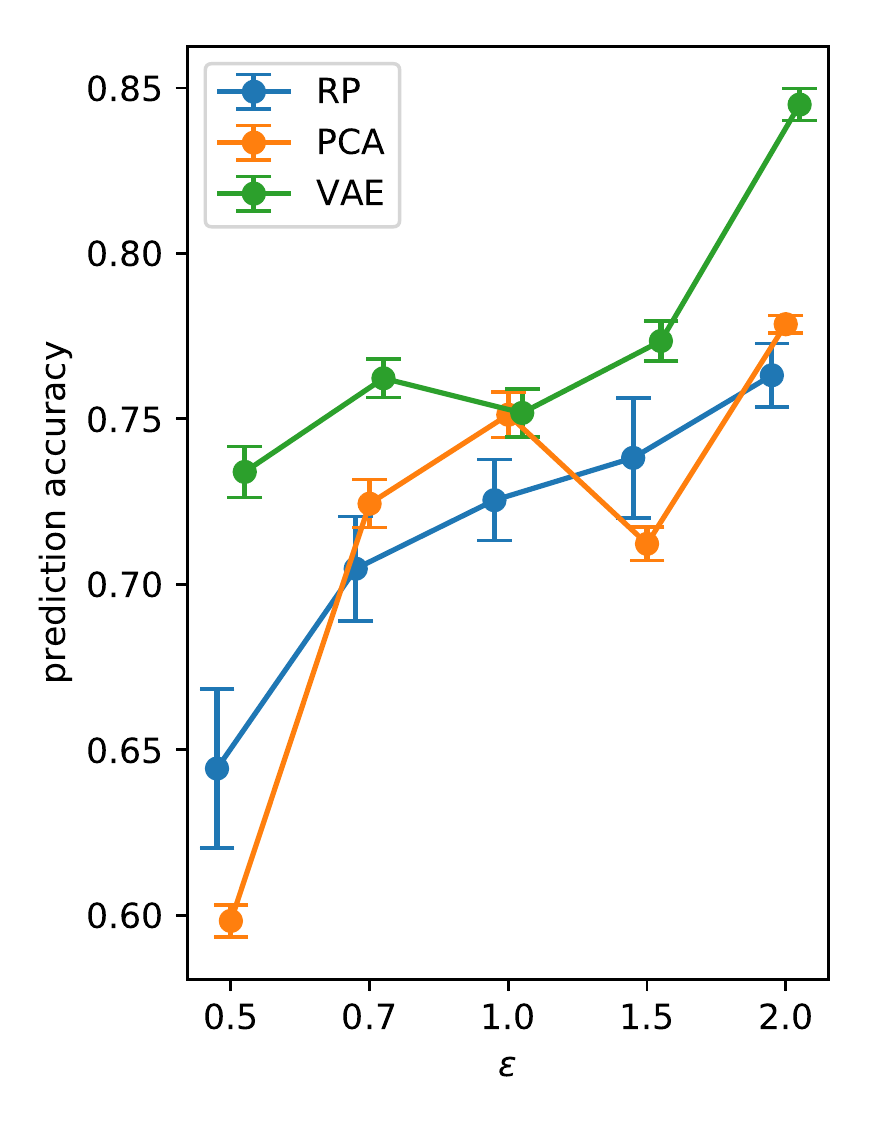}}
\caption{Logistic regression classification accuracy in cancer type classification as a function of $\epsilon$ for case 1. The error bars denote the standard error of the mean when repeating the DP learning but do not cover the uncertainty from hyperparameter selection.}
\label{fig:cancertype_acc_eps}
\end{figure}

\begin{table}[t]
\centering
\caption{Selected hyperparameters for different values of $\epsilon$ in the case 1 of cancer type classification. See Table~\ref{tbl:hyperparam} for explanation of the columns.}
\begin{tabular}{ c | c c c c c c }
  \cline{2-7}
   & \multicolumn{1}{c|}{RP} & \multicolumn{1}{c|}{PCA} & \multicolumn{4}{c|}{VAE} \\
  \cline{2-7}
  $\epsilon$ & \multicolumn{3}{c|}{repr-dim} & \multicolumn{1}{c|}{log-lr} & \multicolumn{1}{c|}{layers} & \multicolumn{1}{c|}{layer-dim} \\
  \hline
  0.5 & 5 & 5 & 5 & -4.6 & 1 & 725 \\
  0.7 & 6 & 14 & 5 & -4.6 & 1 & 880 \\
  1 & 14 & 14 & 5 & -4.1 & 1 & 395 \\
  1.5 & 9 & 9 & 5 & -4.9 & 1 & 1570 \\
  2 & 10 & 11 & 10 & -4.0 & 1 & 680 \\
\end{tabular}

\label{tbl:hyperparam_eps}
\end{table}

\paragraph{Results}

Figure~\ref{fig:cancertype_acc} shows the final prediction accuracy in
the selected 8 cases for $\epsilon = 1$. While none of the methods
fully dominates the others, VAE seems to get some edge, being clearly
the best in about half of the cases and doing decent job in the rest
of the cases too.  The selected hyperparameters are listed in
Table~\ref{tbl:hyperparam}. Interestingly, VAE seems to always end up
with lower dimensionality of the representation than the other two
methods. This could be due to the fact that VAE allows nonlinear
transformations which can help to compress the relevant information in
the data into a smaller number of dimensions. On the other hand, it is
not clear why RP also always chooses lower dimension than PCA.

The prediction accuracy as a function of $\epsilon$ in the case 1 is
shown in Figure~\ref{fig:cancertype_acc_eps} and the corresponding
hyperparameters are shown in Table~\ref{tbl:hyperparam_eps}. As
expected, larger $\epsilon$ results in better accuracy. There is some
variability compared to case 1 in Figure~\ref{fig:cancertype_acc},
which is mostly likely due to the results having been computed with
different hyperparameters. Due to the high computational cost,
variability due to hyperparameter adaptation is not included in the
error bars.

\begin{figure}[t]
\centerline{\includegraphics{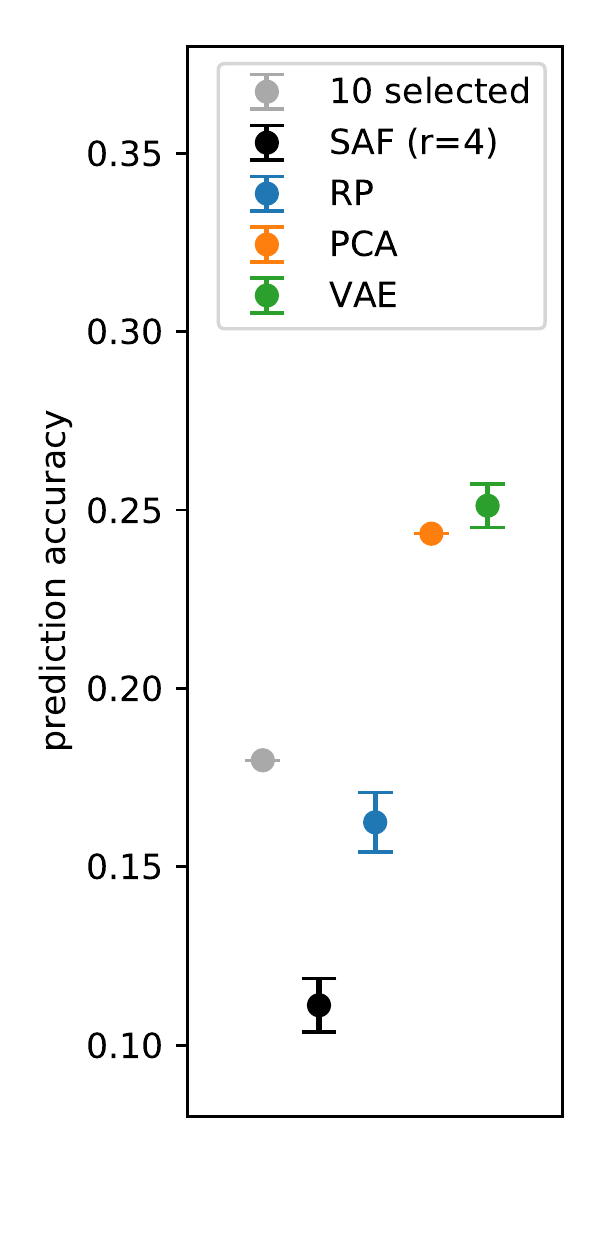}}
\caption{The accuracy of drug sensitivity prediction (Spearman's rank correlation coefficient between the measured ranking of the cell lines and the ranking predicted by the models) with differentially private linear regression ($\epsilon = 1.0$) on the GDSC data.}
\label{fig:drugsens_acc}
\end{figure}

\subsection{Representation learning for DP drug sensitivity prediction}

Our main learning task is to predict the sensitivities of cancer
cell lines to certain drugs. In this task we use data from the
Genomics of Drug Sensitivity in Cancer (GDSC) project \citep{yang2013}
as private data. After preprocessing the data contains 985 samples,
11714 genes and 265 drugs. The data are sparse in the sense that not all
drugs have been tested on all samples. For prediction we use the
differentially private Bayesian linear regression algorithm by
\citet{honkela2018}. The DP linear regression is applied for each drug
separately, using the full $\epsilon$ budget as if it was the only drug
we are interested in.  We then measure and report the average
prediction accuracy over all drugs.

As public data we use the gene expression measurements from the TCGA
data with cancer type classification as the auxiliary prediction
task. The private and public datasets are unified by removing genes
not appearing in both datasets. In addition, since the TCGA gene
expression data are RNA-seq-based while GDSC data are based on
microarrays, we apply quantile normalisation to each gene in the TCGA
data to make it match the distribution of the gene in the GDSC
data. (While this operation theoretically breaks the privacy
guarantees, in practice we can avoid the issue by assuming that the
expression distributions obtained with the microarray technology are
public knowledge.)

\paragraph{Methods}
In addition to RP, PCA and VAE, we also compare to DP feature
selection by Sample and Aggregate framework (SAF) as presented by
\citet{kifer2012}, as well as to using a set of 10 preselected genes
that were used by \citet{honkela2018} in the same prediction task. In
the case of SAF half of the privacy budget is reserved for feature
selection.

RP, PCA and VAE learning was performed in a similar manner as in the
cancer type classification task. For selecting the size of the
representation of SAF, we simply ran it with all possible sizes and
select the best result (which is obviously unfair for the other
methods and would yield a weaker privacy guarantee).

\paragraph{Results}

Figure~\ref{fig:drugsens_acc} shows the average prediction
performance, measured by Spearman's rank correlation. Here PCA and VAE
are the best by some margin, both improving significantly over the
previous state-of-the-art with preselected genes.  On the other hand,
SAF is clearly the worst as the DP feature selection is essentially
random due to small privacy budget, and since it leaves only half of
the privacy budget for the main prediction task.

\section{Discussion}

Our results clearly demonstrate that representation learning with
public data can significantly improve the accuracy of differentially
private learning, compared to using a set of preselected dimensions or
doing differentially private feature selection. Whether it is
beneficial to use more advanced representation learning methods such
as variational autoencoders instead of simple methods such as PCA or
random projections, depends on the task. On some tasks that certainly
seems to be the case.

In our current approach, the representation is learned in an
unsupervised manner and the auxiliary supervised task is only used for
hyperparameter selection. A natural question, that we leave for further
work, is whether representation learning would also benefit from
having an integral auxiliary prediction task that would be learned
concurrently with the representation. The optimisation
target would in that case be a combination of unsupervised
reconstruction error and supervised prediction error. This approach
would require an auxiliary target variable, as is the case in this
work with hyperparameter optimisation.

In general, we believe DP learning can be important in opening genomic
and other biomedical datasets to broader use. This can significantly
advance open science and open data, and lead to more accurate models
for precision medicine. So far, the accuracy of DP learning in most
practical applications is not comparable to realistic non-private
alternatives. Our present work makes an important contribution toward
making DP learning practical.

In the present work the representation learning was not performed
under DP. This is a clear limitation if the other data set also needs
privacy protection. This can in theory be addressed easily, by simply
training the representation model under DP, but this will likely have
an impact on the accuracy of the final model. Ultimately we believe
that a clever combination of private and non-private data such as in
our paper can lead to the best results.

\section*{Acknowledgements}

We thank the reviewers of an earlier workshop version of this article for helpful comments.\vspace*{-12pt}

\section*{Funding}

This work has been supported by the Academy of Finland
[Finnish Center for Artificial Intelligence FCAI and
grants 292334, 294238, 303815, 303816, 313124].

\bibliographystyle{natbib}

\bibliography{paper.bib}

\end{document}